\documentclass[letterpaper, 10pt, conference]{ieeeconf}

\IEEEoverridecommandlockouts

\usepackage{pdfpages}
\usepackage[active]{srcltx}
\usepackage{algorithm}
\usepackage{algpseudocode}
\usepackage{verbatim}
\usepackage{epsfig}
\usepackage{amsmath}
\usepackage{xcolor}
\usepackage{amssymb}
\usepackage{psfrag}
\usepackage[T1]{fontenc}
\usepackage{graphicx}
\usepackage[usenames,dvipsnames]{pstricks}
\usepackage{pst-grad} % For gradients
\usepackage{pst-plot} % For axes

\graphicspath{{./figure/}}

\newtheorem{theorem}{Theorem}[section]
\newtheorem{lemma}{Lemma}[section]

\newtheorem{remark}{Remark}

\newcommand{\Eps}{\mathcal{E}}

\definecolor{color1}{rgb}{0,0,0}

\title{\LARGE \bf
Exact Topology Learning in a Network of Cyclostationary Processes 
}

\author{Harish Doddi$^1$, Saurav Talukdar$^{1}$, Deepjyoti Deka$^{2}$ and Murti Salapaka$^{3}$ \\% <-this % stops a space
\thanks{$^{1}$Harish Doddi and Saurav Talukdar are with Department of Mechanical Engineering, University of Minnesota, Minneapolis, USA,
{\tt\small doddi003@umn.edu, sauravtalukdar@umn.edu}}%
\thanks{$^{2}$Deepjyoti Deka is with Theory Division, Los Alamos National Laboratory, Los Alamos, USA,
{\tt\small deepjyoti@lanl.gov}}%
\thanks{$^{3}$Muti V. Salapaka are with Department of Electrical and Computer Engineering, University of Minnesota, Minneapolis, USA,
{\tt\small murtis@umn.edu}}
}
\begin{document}
\maketitle
\thispagestyle{empty}
\pagestyle{empty}

%%%%%%%%%%%%%%%%%%%%%%%%%%%%%%%%%%%%%%%%%%%%%%%%%%%%%%%%%%%%%%%%%%%%%%%%%%%%%%%%
\begin{abstract}
Learning the structure of a network from time-series data, in particular cyclostationary data, is of significant interest in many disciplines such as power grids, biology, finance. In this article, an algorithm is presented for reconstruction of the topology of a network of cyclostationary processes. To the best of our knowledge, this is the first work to guarantee exact recovery without any assumptions on the underlying structure. The method is based on a lifting technique by which cyclostationary processes are mapped to vector wide sense stationary processes and further on semi-definite properties of matrix Wiener filters for the said processes. We demonstrate the performance of the proposed algorithm on a Resistor-Capacitor network and present the accuracy of reconstruction for varying sample sizes. 
\end{abstract}

%%%%%%%%%%%%%%%%%%%%%%%%%%%%%%%%%%%%%%%%%%%%%%%%%%%%%%%%%%%%%%%%%%%%%%%%%%%%%%%%
\section{INTRODUCTION}
Networks are a framework used extensively for analysis of the behavior of complex systems like the brain \cite{brainnetwork}, climate \cite{KimNorth1997}, gene regulation dynamics \cite{geneNetwork}, disease spread \cite{diseaseSpread}, power grid \cite{deka2017structure} and others. Network representations help identify essential influence pathways in a complex system, thereby enhance the understanding of the dynamics of complex systems. A problem of interest to multiple communities like control theory, machine learning and signal processing is the inference of the influence pathways among the entities of interest from observation of the entities, which is sometimes referred as structure learning or topology learning \cite{gonccalves2008necessary, drton2017structure, giannakis2018topology}. For example: given a time series collection of stock prices over a time horizon, it is of interest to infer the influence paths among the collection of stocks from the observed time series of stock prices. 

Fundamentally, there exist two different approaches for inference of the influence pathways among the entities from observations. The first is active learning, where a particular entity is perturbed by an external agent and its influence on the rest of the entities is examined \cite{nabi2014network}. However, the disadvantage of such an approach is that it is not always possible to actively excite a specific entity to glean the network structure. For example: it is not always possible to turn off or change generator set points to infer the structure of a power distribution network. The second approach to structure learning involves passive or non invasive approaches, where, the system is not perturbed actively and topology is inferred solely from the observations of the variables of interest \cite{Pea88},\cite{pearl2009causality}, \cite{meinshausen2006high},\cite{friedman2008sparse}. In this article we focus on the passive approach to topology inference. 

Inference of the network structure from observations assuming that the entities are a collection of random variables is studied in \cite{Pea88, friedman2008sparse}. However, such approaches are ineffective when lagged correlations are significant and the past of one entity influences the present of another entity. In such situations, the observed entities are modeled as stochastic processes or time series. In the time series setting, topology inference under the assumption of wide sense stationary time series using multivariate Wiener filtering is described in \cite{materassi2012problem} and using power spectral density in \cite{shahrampour2015topology}. These approaches are not applicable to the case of non stationary time series. In this article we utilize multivariate Wiener filtering for topology inference with consistency guarantees for non stationary time series, which are cyclostationary. Cyclostationary processes are characterized by a periodic mean and correlation function. Many phenomenon in nature exhibit periodic behaviour \cite{napolitano2012generalizations}; examples include weather patterns \cite{kim1996}\cite{KimNorth1997}, regulatory processes in human body \cite{biologyRef1}, \cite{biologyRef2}, \cite{biologyRef3}, trends in stock market \cite{stockRef1}, \cite{stockRef2}, motion of mechanical systems \cite{mechRef1}, \cite{mechRef2}, communication systems \cite{communicationRef1}, \cite{communicationRef2}, and planet movements \cite{planetRef1}. 
 In \cite{talukRecon2015}, the authors use multivariate Wiener filtering for topology inference in a collection of cyclostationary time series. However, the approach presented in \cite{talukRecon2015} introduces possibly many spurious links and does not provide exact inference of the underlying network topology from the observations. 
In this article, we extend the work in \cite{talukRecon2015} and present an algorithm with guarantees for {\it exact} reconstruction of a network structure in a collection of cyclostationary processes. Our approach does not require any structural assumptions on the underlying graph structure nor any knowledge of the system parameters. Our results are based on the properties of the Wiener filters of collection of wide sense stationary time series for inference of exact topology in a network of wide sense stationary processes \cite{harish2018topology, talukdar2017learning}. We validate the algorithm presented through implementation on time series of nodal states generated from an interconnected Resistor-Capacitor (RC) network with input cyclostationary process. We compare the performance of our algorithm against the algorithms presented in \cite{talukRecon2015} and show the superiority of the approach presented here over prior work. To the authors' best knowledge, this is the first work on exact topology reconstruction of cyclostationary processes with provable guarantees. 

The rest of the paper is organized as follows. Section II describes the learning problem of a general linear dynamical system with cyclostationary inputs. In section III, the approach for tackling the topology learning problem from cyclostationary time series data is presented. Analytic results as well as an implementable algorithm is presented. The performance of the algorithm is demonstrated on a simulated RC network in Section IV. Conclusions are presented in Section V.

\subsection*{\bf Notation:}
\noindent 
The symbol $:=$ denotes a definition\\
\begin{comment}
$\|x\|$ : $2$ norm of a vector $x$ \\
$[x]_i$ : $i-th$ element of a vector $x$ \\
$x(k)$ : $k^{th}$ sample of time domain signal $x$ \\
$\mathsf{x}(z)$ : z-transform of $x(k)$ \\
$[A]_{j,:}$ : $j-th$ row of matrix $A$\\
$[A]_{:,i}$ : $i-th$ column of matrix $A$\\
$[A]_{i,j}$ : $i-th$ row, $j-th$ column element of matrix $A$\\
\end{comment}
$\boldsymbol{0}$ : zero matrix of appropriate dimension\\ 
$A^{*}$ : the conjugate transpose of matrix $A$\\
$A'$ : transpose of a matrix $A$\\
$\|A\|_{\infty}$ : largest row sum of matrix $A$\\
$A \succeq 0$ : $A$ is a positive semi-definite matrix\\
$A \preceq 0$ : $A$ is a negative semi-definite matrix\\

%%%%%%%%%%%%%%%%%%%%%%%%%%%%%%%%%%%%%%%%%%%%%%%%%%%%%%%%%%%%%%%%%%%%%%%%%%%%%%%%

\section{Linear Dynamical System with cyclostationary inputs}
In this section, we briefly define and describe a networked cyclostationary process. 
\subsection{Introduction to Cyclostationary Process}
A random process $x(t) \in L^{2}(\Omega,\mathcal{F},\mathcal{P})$, with its mean function $m(t) := \mathbb{E}[x(t)]$ and correlation function $R_x(s,t) := \mathbb{E}[x(s)x(t)]$ is wide sense cyclostationary (WSCS) or periodically correlated with period $T>0$ if for every $s,t \in \mathbb{Z}$ there is no value smaller than $T$ such that $m(t) = \mathbb{E}[x(t)]= \mathbb{E}[x(t+T)]= m(t + T)$ and $R_{x}(s,t) = \mathbb{E}[x(s)x(t)] = \mathbb{E}[x(s+T)x(t+T)] = R_{x}(s+T,t+T)$. The random processes $x(t)$ and $e(t)$ are said to be \textbf{jointly wide sense cyclostationary (JWSCS)} if $x(t)$ and $e(t)$ are cyclostationary with period $T$ and the cross correlation function $R_{x,e}(s,t)$ is periodic with period $T$ that is $R_{x,e}(s,t)$ $=R_{x,e}(s+T,t+T)$. An example of a cyclostationary process is a periodic signal superimposed with a \textbf{wide sense stationary (WSS)} signal. A WSS signal by itself is also a cyclostationary signal with period $T=1$.

\subsection{Network of Cyclostationary Processes}
Consider a system with $m$ sub-systems indexed by $i$ with the dynamics of each sub-system $x_i$ given by 
\begin{align}\label{eqn:lindyn}
 \sum_{n=1}^{l}a_{n,i}\frac{d^nx_i}{dt^n}=\sum_{j=1,j\neq i}^{m}b_{ij}(x_j(t) - x_i(t)) + p_i(t), 
\end{align}
$ i \in\{1,2,..,m\}$, where $x_i(t)$ is the time series data of the sub-system $i$; $a_{m,i}$ and $b_{ij}$ are system parameters and $p_i(t)$ is a forcing function acting on sub-system $i$, which is a zero mean WSCS of period $T$, uncorrelated with $\{p_j(t)\}_{i \neq j}$. Since a linear transformation of $\{p_j(k)\}_{j=1}^m$ results in $\{x_j(k)\}_{j=1}^m$, the collection of $(\{x_j(k)\}_{j=1}^m,\{p_j(k)\}_{j=1}^m)$ are JWSCS \cite{talukRecon2015}. All the processes and the system parameters of (\ref{eqn:lindyn}) are real valued and $b_{ij}$ are non-negative. Equation (\ref{eqn:lindyn}) is known as Linear Dynamical Model (LDM).

Converting (\ref{eqn:lindyn}) from time domain to z-domain by applying bilinear transform (Tustin's method \cite{oppenheim1999discrete}) on (\ref{eqn:lindyn}), we obtain 
\small
\begin{align}\label{eqn:cyclo_LDM}
\mathsf{S_i}(z)\mathsf{x}_i(z) &= \sum_{j=1,j\neq i}^{m}b_{ij}\mathsf{x}_j(z) + \mathsf{p}_i(z),\nonumber\\
\mathsf{x}_i(z) &= \sum_{j=1,j\neq i}^{m}\frac{b_{ij}}{\mathsf{S_i}(z)}\mathsf{x}_j(z) + \frac{\mathsf{p}_i(z)}{\mathsf{S_i}(z)},\nonumber\\
\mathsf{x}_i(z) &= \sum_{j=1}^{m}\mathsf{h}_{ij}(z)\mathsf{x}_j(z) + \mathsf{e}_i(z)\\
{x}_i(k) &= \sum_{j=1}^{m}\sum_{n= -\infty}^{\infty}h_{ij}(n){x}_j(k-n) + {e}_i(k)
\end{align}
\normalsize
where, $\mathsf{h}_{ii}(z)=0$, $\mathsf{S_i}(z)=\sum_{n=1}^{l}a_{n,i}(\frac{2(1-z^{-1})}{\Delta t(1+z^{-1})})^{n}+\sum_{j=1,j\neq i}^{m} b_{ij}$ is a scalar function, $\mathsf{h}_{ij}(z) = \frac{b_{ij}}{\mathsf{S_i(z)}}$, $\mathsf{e}_i(z)=$ $\frac{\mathsf{p}_i(z)}{\mathsf{S_i(z)}}$. Note that $\mathsf{h}_{ij}(z)$ is the Z-transform of $h_{ij}(n)$, that is $\mathsf{h}_{ij}(z) = \mathcal{Z}[h_{ij}(n)]$. Here, $\mathsf{x}_i(z)$, $\mathsf{p}_i(z)$, $\mathsf{e}_i(z)$ are the z-transforms of $x_i(k),p_i(k) \text{ and } e_i(k)$ respectively. $e_j(k)$ is the input noise which is zero mean WSCS and uncorrelated with $\{e_i\}_{i=1,i \neq j}^{m}$.

The Linear Dynamic Graph (LDG) given by $\mathcal{G}(V,E)$ associated with the LDM described by (\ref{eqn:lindyn}) is obtained by setting each vertex (node) $i \in V$ to represent every subsystem $i$ and by placing a direct edge in $E$ from vertex $i$ to vertex $j$ if $b_{ji}\not=0$. The set of {\it children} of node $x_j$ is defined as $\mathcal{C}_G(x_j):=\{x_i| (x_j,x_i)\in E\}$, its set of parents as $\mathcal{P}_G(x_j):=\{x_i| (x_i,x_j)\in E\}$ and its set of kins as $\mathcal{K}_G(x_j):=\{x_i| x_i\neq x_j \text{ and } x_i \in \mathcal{C}_G(x_j)~\cup\mathcal{P}_G(x_j)~\cup~\mathcal{P}_G(\mathcal{C}_G(x_j)) \}.$
The kin topology/moral graph of the directed graph $\mathcal{G}(V,E)$ is defined as an undirected graph $\mathcal{G}_M(V,\Tilde{E})$ with a vertex set $V$ and an edge set $\Tilde{E} = \{(x_i,x_j)|x_i \in \mathcal{K}_G(x_j) \forall j\}$. The topology of the directed graph $\mathcal{G}(V,E)$ is defined as an undirected graph $\mathcal{G}_T(V,\Tilde{E})$ with a vertex set $V$ and an edge set $\Tilde{E} = \{(x_i,x_j)|(x_i,x_j) \text{ or } (x_j,x_i) \in E\}$ and is denoted by $top(\mathcal{G})= \mathcal{G}_T$. Example of topology and kin topology are shown in Fig. \ref{fig:rcnetwork}(c) and (d) respectively. In the topology $\mathcal{G}_T(V,\Tilde{E})$ of the directed graph $\mathcal{G}$, a neighbor is defined as $\mathcal{N}_{\mathcal{G}_T}(x_j):=\{x_i| (x_i,x_j) \in \Tilde{E}\}$ and a two hop neighbor is defined as $\mathcal{N}_{\mathcal{G}_T}(x_j,2):=\{x_i\in V| (x_j,x_k),(x_i,x_k) \in \Tilde{E} \text{ for some } x_k \in V\}$

\begin{figure}[tb]
	\centering
	\begin{tabular}{cc}
	\includegraphics[width=0.4\columnwidth]{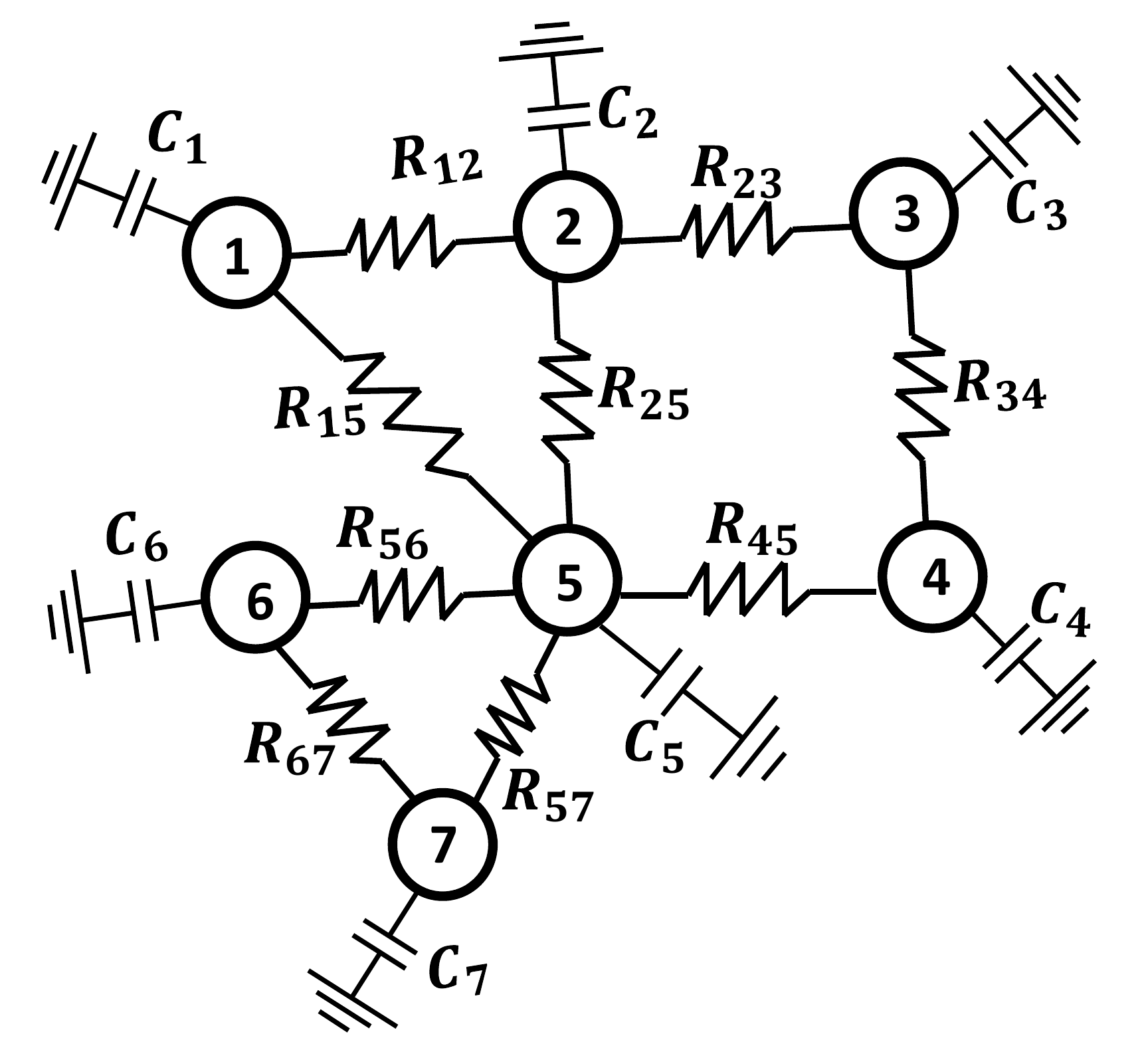}&
	\includegraphics[width=0.35\columnwidth]{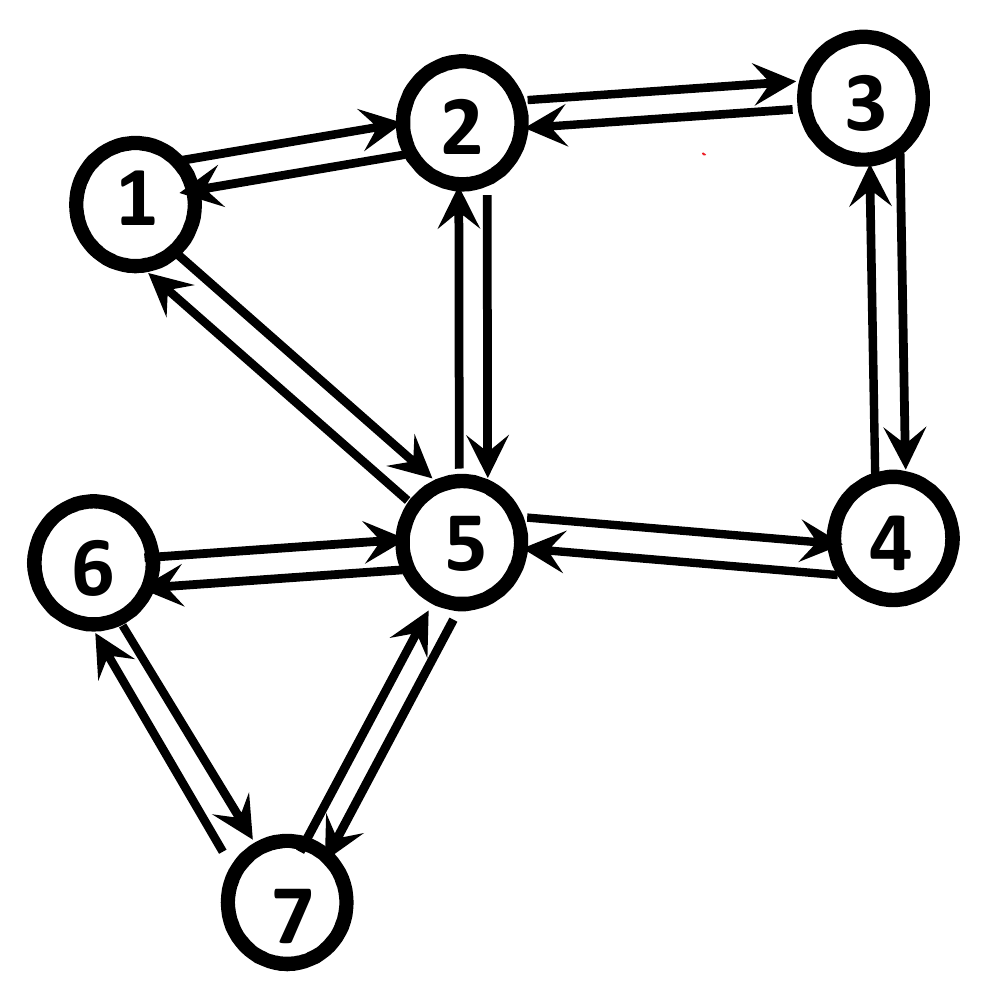}\\
		(a) & (b)
	\end{tabular}
	\begin{tabular}{cc}
	\includegraphics[width=0.32\columnwidth]{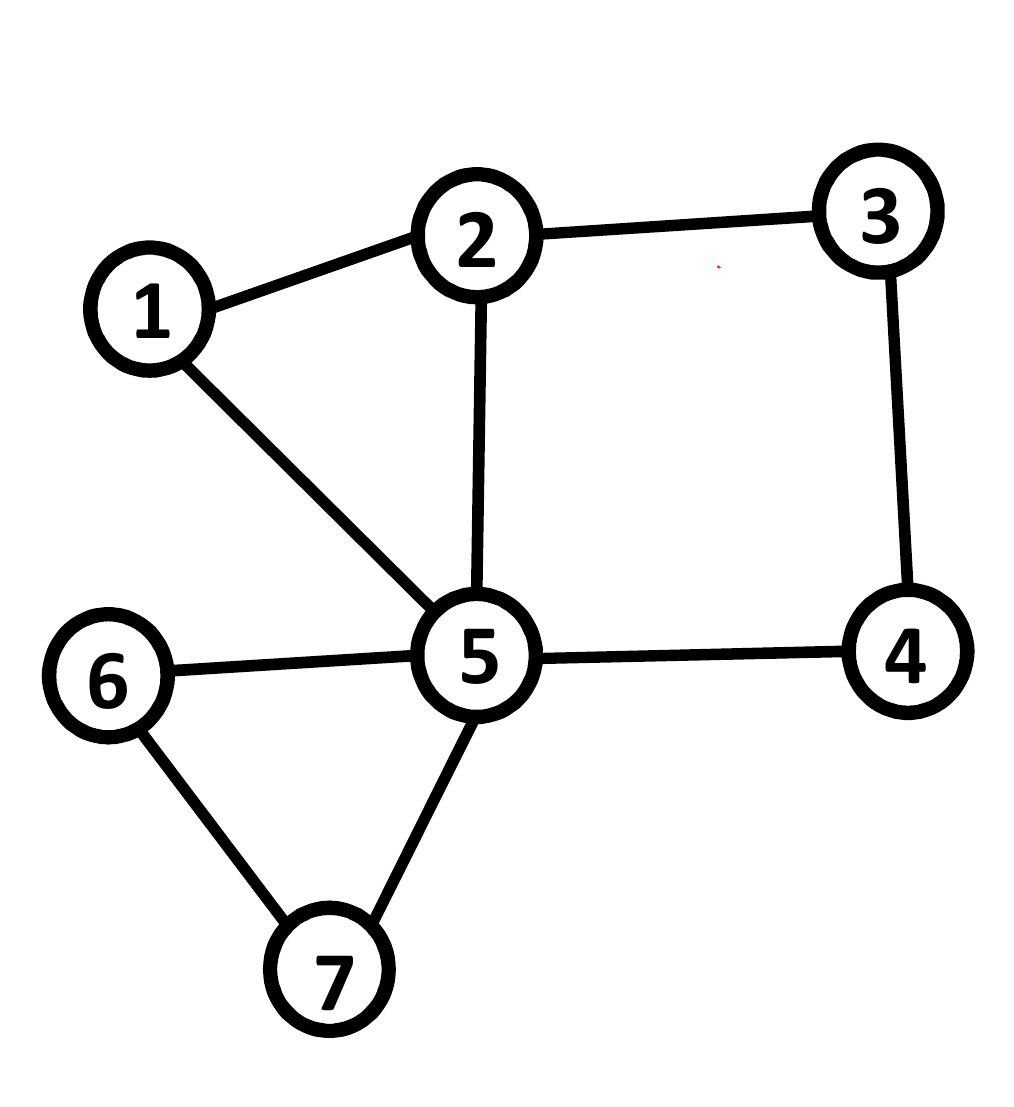}&
	\includegraphics[width=0.35\columnwidth]{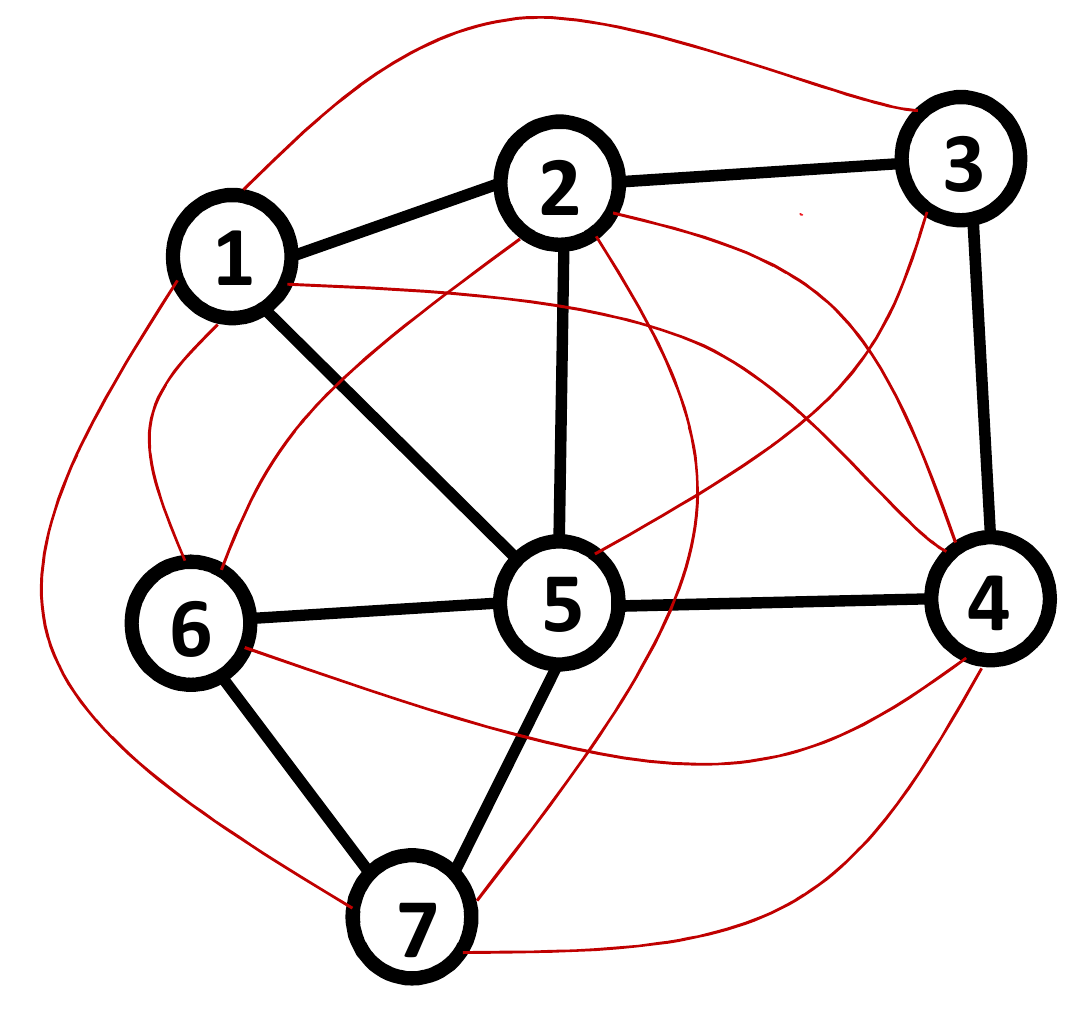}\\
		(c) & (d)
	\end{tabular}
	\caption{(a) RC network, (b) directed graph, (c) it's topology and (d) kin topology/moral graph.
		\label{fig:rcnetwork}}
\end{figure}

\subsection{Lifting to vector wide sense stationary processes}

{The vector random process $X_i(k)$ is said to be (weakly) stationary if $E([X_i(k)]_j) = m_i(j)$, where $[X_i(k)]_j$ is the $j^{th}$ element of $X_i(k)$ and $R_{X_i}^{jk}(s,t) = E([X_i(s)]_j[X_i(t)]_k)=E(x_i(s+j-1)x_i(t+k-1)) = E(x_i(s-t+j)x_i(k)) = R_{X_i}^{jk}(s-t)$ for all $s, t \in \mathbb{Z}, \text{ and } j, k \in \lbrace1,2,\cdots,T\rbrace$}

Any cyclostationary process with a period $T$ can be represented as multivariate, vector, wide sense stationary process through the {\it lifting} process \cite{hurd2007periodically} described next. Processes $x_i(k)$ and $e_i(k)$ are lifted to $T$-variate WSS processes $X_i(k)= [x_i(kT)\ ...\ x_i(kT+T-1)]'$ and $E_i(k)= [e_i(kT)\ ...\ e_i(kT+T-1)]'$ respectively and their z-transform's are denoted as $\mathsf{X}_{i}(z)$ and $\mathsf{E}_{i}(z)$ respectively. The vector process $E_i(k)$ is uncorrelated with $\{E_j(k)\}_{j\neq i}$.

\begin{lemma}
Given the dynamics of a network of cyclostationary processes $\{x_i(k)\}_{i=1}^m$ in (\ref{eqn:cyclo_LDM}), the dynamics of the lifted vector WSS processes $\{X_i(k)\}_{i=1}^m$ is given by $\mathsf{X}_{j}(z) = \sum_{i=1}^{m}\mathsf{H}_{ji}(z)\mathsf{X}_i(z) + \mathsf{E}_j(z)$, where $\mathsf{H}_{ji}(z) = D(z^{\frac{1}{T}})\mathsf{h}_{ji}(z^{\frac{1}{T}})$
\end{lemma}
\begin{proof}
\textcolor{red}{See the appendix}
\end{proof}

\

The dynamics of $\{X_j\}_{j=1}^{m}$ is given by 
\small
\begin{align}\label{dynamiclink}
X_j(k) &= \sum_{i=1}^{m}H_{ji}*X_i(k)+E_j(k) \\
\mathsf{X}_{j}(z) &= \sum_{i=1}^{m}\mathsf{H}_{ji}(z)\mathsf{X}_i(z) + \mathsf{E}_j(z)\\
\text{ Where, }\nonumber \\
\mathsf{H}_{ji}(z) &= \mathcal{Z}[H_{ji}(k)] = D(z^{\frac{1}{T}}) \mathsf{h}_{ji}(z^{\frac{1}{T}})\nonumber \\
E_i(k) = &\begin{bmatrix}e_i(kT)\\e_i(kT+1)\\\vdots\\e_i(kT+T-1) \end{bmatrix};X_i(k) = \begin{bmatrix}x_i(kT)\\x_i(kT+1)\\\vdots\\x_i(kT+T-1) \end{bmatrix}\nonumber\\
\mathsf{E}_i(z) = &\begin{bmatrix}\mathsf{e}_i(z^{\frac{1}{T}})\\z^{\frac{1}{T}}\mathsf{e}_i(z^{\frac{1}{T}})\\\vdots\\z^{\frac{T-1}{T}}\mathsf{e}_i(z^{\frac{1}{T}}) \end{bmatrix}= R(z^{\frac{1}{T}})\mathsf{e}_i(z^{\frac{1}{T}}) = R(z^{\frac{1}{T}})\frac{\mathsf{p}_i(z^{\frac{1}{T}})}{\mathsf{S_i}(z^{\frac{1}{T}})}\nonumber\\
\mathsf{X}_i(z) = &\begin{bmatrix}\mathsf{x}_i(z^{\frac{1}{T}})\\z^{\frac{1}{T}}\mathsf{x}_i(z^{\frac{1}{T}})\\\vdots\\z^{\frac{T-1}{T}}\mathsf{x}_i(z^{\frac{1}{T}}) \end{bmatrix}=R(z^{\frac{1}{T}})\mathsf{x}_i(z^{\frac{1}{T}})\nonumber\\
D(z) = &\begin{bmatrix}z^0&z^1\dots&z^{T-1}\\\vdots&\ddots&\vdots\\z^{-(T-1)}&\dots&z^0\\\end{bmatrix};R(z)= \begin{bmatrix}1\\z\\\vdots\\z^{T-1}\end{bmatrix}\nonumber
\end{align}
\normalsize
Here, $\mathsf{H}_{ji}(z)$ is a $T \times T$ transfer matrix, $\mathsf{E}_i(z)=\mathcal{Z}(E_i(k))$ and $\mathsf{X}_i(z)=\mathcal{Z}(X_i(k))$. The collection $\{X_i(k)\}_{i=1}^m$ are jointly vector WSS processes. {The LDG for the vector processes $\{X_i(k)\}_{i=1}^m$ is defined as follows: a directed edge from $X_i$ to $X_j$ exists if and only if $\mathsf{h}_{ji} \neq 0$.
}

\begin{lemma}
A directed edge from a cyclostationary process $x_i$ to $x_j$ in the LDG of (\ref{eqn:lindyn}) exists if and only if there exists a directed edge from $X_i$ to $X_j$ 
\end{lemma}
\begin{proof}
Given that a directed edge exists from cyclostationary process $x_i$ to $x_j$ then $\mathsf{h}_{ji}(z)\neq 0$ in (\ref{eqn:cyclo_LDM}). It follows from (\ref{dynamiclink}) that $\mathsf{H}_{ji}(z)\neq \boldsymbol{0}$ which implies that there is a directed edge from vector stationary process $X_i$ to $X_j$. 
To prove the converse, if suppose there is a directed edge from $X_i$ to $X_j$ then $\mathsf{H}_{ji}(z) \neq \boldsymbol{0}$ in (\ref{dynamiclink}). This implies that $\mathsf{h}_{ji}(z)\neq 0$ from (\ref{dynamiclink}) and there is a directed edge from $x_i$ to $x_j$.
\end{proof}
\begin{remark}
The above lemma concludes that kin relationship's in the LDG of $\{x_j\}_{j=1}^{m}$ identical to the kin relationship's in the LDG of $\{X_j\}_{j=1}^{m}$. This enables us to reconstruct the topology of LDG for the cyclostationary processes $\{x_j\}_{j=1}^{m}$ by reconstructing the topology for their equivalent vector stationary processes $\{X_j\}_{j=1}^{m}$.
\end{remark}
The Power Spectral Density (PSD) of the vector WSS process $X_i(k)$ is related to that of scalar WSCS process $x_i(k)$. The PSD matrix ${\Phi}_{P_i}(z), {\Phi}_{E_i}(z)$ are block matrices. Note that $\mathsf{P}_i(z) = \mathsf{S_i}(z^{\frac{1}{T}})\mathsf{E}_i(z)$. Therefore, 
\begin{align}
{{\Phi}_{\mathsf{P}_i}(z)}=& \mathsf{S_i}(z^{\frac{1}{T}}){{\Phi}_{\mathsf{E}_i}(z)}(\mathsf{S_i}(z^{\frac{1}{T}}))^*\nonumber\\
\Rightarrow
{\Phi}_{\mathsf{E}_i}^{-1}(z)=& {\Phi}_{\mathsf{P}_i}^{-1}(z)|\mathsf{S_i}(z^{\frac{1}{T}})|^2\label{phiE_phiP}
\end{align}

The rest of the article focuses on the reconstruction of topology of a dynamical network of $T\times 1$ vector stationary processes that are jointly wide sense stationary (in contrast to $x_i$ which are jointly wide sense cyclostationary). 

\section{Learning the structure from time series data}
Here we assume that $x_i(k)$ is available as a measured time series for all subsystems $i$. Given the time series data of the dynamical system with noise modeled as cyclostationary, the main aim of this article is to reconstruct the topology of the LDG of the system. The time series $x_i(k)$ is lifted to vector WSS process $X_i(k)$.
The output processes $\{X_j(k)\}_{j=1}^{m}$ of (\ref{dynamiclink}) are defined in a compact way as:
\small
\begin{align}\label{eq:LDG general}
X(k)=&\mathbb{H}*X(k) + E(k)\\
\mathsf{X}(z)=&\mathbb{H}(z)\mathsf{X}(z) + \mathsf{E}(z),\text{ where }\\
X(k)=\begin{bmatrix}X_1(k)\\\vdots\\X_m(k)\end{bmatrix};\ &E(k)=\begin{bmatrix}E_1(k)\\\vdots\\E_m(k)\end{bmatrix};\nonumber\\
\mathsf{X}(z)=\mathcal{Z}(X(k)); \ &\mathsf{E}(z)=\mathcal{Z}(E(k))\nonumber
\end{align}
\normalsize
The transfer block matrix $\mathbb{H}(z)$ is $mT \times mT$ matrix with size of each block $T \times T$ and its diagonal entries are $\boldsymbol{0_{T \times T}}$. The matrix $\mathbb{H}(z)$ in (\ref{eq:LDG general}) characterizes the interconnection between the processes $X_i(k)$ and $X_j(k)$ with the $(i,j)^{th}$ block of $\mathbb{H}(z)$ being $\mathsf{H}_{ij}(z)$. The relation (\ref{eq:LDG general}) defines a map from a vector processes $E(k)$ to a vector processes $X(k)$. The LDM described by (\ref{eq:LDG general}) is well-posed if the operator $(\mathbb{I}-\mathbb{H}(z))$ is invertible almost surely. Therefore for any vector processes ${E}(k)$ there exists a vector ${X}(k)$ of processes and the LDM described by (\ref{eq:LDG general}) is said to be topologically detectable if $\Phi_{E_j} \succ 0$ for any $\omega \in[-\pi,\pi]$ and for any $j=1,...,m$.

 The proposed algorithm in this article for exact reconstruction of the true topology is based on inverse PSD of the time series described above.

\par 

\subsection{Reconstruction of Moral Graph}
In this section, we present an algorithm for reconstruction of the network structure of LDG $\mathcal{G}$ of (\ref{dynamiclink}), which is well-posed and topologically detectable. 

\begin{comment}
Consider any two vertices $X_i,X_j$ with a solution $\boldsymbol{\mathsf{W_{ji}}}(z)$ of (\ref{eq:vec_opt_freq}). If $\boldsymbol{\mathsf{W_{ji}}}(z) \neq \boldsymbol{0}$, add an edge between $X_i,X_j$. The process is repeated for all $j$ and $i.$ . Let the resulting undirected graph be $G$. Then it is shown in \cite{talukRecon2015} that $G$ is the moral graph of the LDG associated with the LDM (\ref{eq:LDG general}) (and thus the LDG associated with the LDM (\ref{cyclo_LDM})). The $(j,i)$ entry block matrix (\cite{talukRecon2015}) of $\Phi_{X}^{-1}(z)$ is 
\small
\begin{align}\label{Phi_wji}
 {\boldsymbol{B}_{j}^{\prime}}{{\Phi}_{X}^{-1}}{\boldsymbol{B}_{i}}=&-\Phi_{\Eps_j}^{-1}\boldsymbol{\mathsf{W}}_{ji}(z)\\
 \text{where } \Eps_j(k)=& X_j(k)- \sum_{i=1,i\neq j}^{m}(\boldsymbol{\mathsf{W}}_{ji}*X_i)(k).\nonumber
\end{align}
\normalsize
\end{comment}

From (\ref{eq:LDG general}) $X=(\mathbb{I}-\mathbb{H}(z))^{-1}E$ which implies that $\Phi_{X}^{-1}=(\mathbb{I}-\mathbb{H})^{*}\Phi_{E}^{-1}(\mathbb{I}-\mathbb{H})$. We now consider the $(i,j)$ block matrix of $\Phi_{X}^{-1}$, for $i\neq j$. Using the block diagonal structure of $\Phi_{E}$, we have
\small
\begin{align}\label{PhiInv_expression}
&{\boldsymbol{B}_{j}^{\prime}}{{\Phi}_{X}^{-1}}{\boldsymbol{B}_{i}}={\boldsymbol{B}_{j}^{\prime}}(\mathbb{I}-\mathbb{H})^{*}\Phi_{E}^{-1}(\mathbb{I}-\mathbb{H}){\boldsymbol{B}_{i}}\nonumber \\
&=-{{\Phi}_{E_j}^{-1}}{\mathsf{H}_{ji}}-{{\mathsf{H}_{ij}}^{*}}{{\Phi}_{E_i}^{-1}}+{\sum_{k=1}^{m}}{(\mathsf{H}_{kj})^{*}}{{\Phi}_{E_k}^{-1}}(\mathsf{H}_{ki}).
\end{align}
\normalsize

where $\boldsymbol{B}_j = [\boldsymbol{0} \ \boldsymbol{0} \ .. \ \boldsymbol{0} \ \boldsymbol{I}_{T\times T}\ \boldsymbol{0}\ ..\ \boldsymbol{0}]^{\prime} $ is a matrix $\in \mathbb{R}^{mT\times T}$ that has an identity matrix $\boldsymbol{I}_{T\times T}$ as the $j^{th}$ block and other blocks as $\boldsymbol{0}$. 

If $(i,j)$ are not kins then $j$ cannot be the parent of $i$ ($\mathsf{H}_{ij} = \boldsymbol{0}$) or child of $i$ ($\mathsf{H}_{ji} = \boldsymbol{0}$) or spouse of $i$ (which implies that there does not exist any $k \in \{1,...,m\} \text{ with } \mathsf{H}_{kj} \neq \boldsymbol{0}\ \text{and}\ \mathsf{H}_{ki} \neq \boldsymbol{0}$). 

\begin{comment}
Therefore from (\ref{PhiInv_expression}), $(j,i)$ entry block matrix of $\Phi_{X}^{-1}(z)$ is zero and thus using (\ref{Phi_wji}) it follows that $\boldsymbol{\mathsf{W}}_{ji}(z)=\boldsymbol{0}$. 
\end{comment}
This gives a sufficient criterion to identify kins which include parents, children, spouses and the moral graph is obtained (see Fig. \ref{fig:rcnetwork}(d)). The limitation with the moral graph is that it may include substantial spurious links between two-hop neighbors as shown in Fig. \ref{fig:rcnetwork}(d). Thus, it is important to devise a method for the detection and removal of spurious links from the moral graph and to deduce the topology of the LDG; which forms the focus of the next section.
\subsection{Pruning the Spurious Strict Two Hop Neighbor Links}
The following theorem will provide a sufficient condition for elimination of all strict two-hop neighbors in the moral graph $\mathcal{G}_M$ which results in true topology ${\mathcal{G}_T}$ 
\begin{theorem}\label{3.1}
Consider a LDM $(\mathbb{H}(z),{E})$ which is well-posed and topologically detectable, with its associated LDG $\mathcal{G}$ and topology $\mathcal{G_T}$. Let the output of the LDM be given by $X(k)$ according to (\ref{eq:LDG general}).
Let the vertices $i$ and $j$ in $\mathcal{G_T}$ be such that, $i \in N_{\mathcal{G_T}}(j,2)$ but $i\notin N_{\mathcal{G_T}}(j)$, that is, $i,j$ are strict two-hop neighbors. Then, $\boldsymbol{\boldsymbol{B}_{j}^{\prime}}{{\Phi}_{X}^{-1}}{\boldsymbol{B}_{i}} \succeq 0$, for all $\omega \in [0, 2 \pi)$.
\end{theorem}

\begin{comment}
\begin{proof}
Given that $i\notin N_{\mathcal{G_T}}(j)$ which implies that ${\mathsf{H}_{ji}(\omega)}={{\mathsf{H^{*}}_{ij}(\omega)}}=\boldsymbol{0}_{T \times T}$. Further, $i \in N_{\mathcal{G_T}}(j,2)$, so $ N_{\mathcal{G_T}}(j,2)\cap N_{\mathcal{G_T}}(j)$ is non-empty. It follows from (\ref{Phi_wji}), (\ref{PhiInv_expression}) through algebraic expansions that $\boldsymbol{\mathsf{W}}_{ji} \preceq 0$, for all $\omega \in [0, 2 \pi)$.
\end{proof}
\end{comment}

\begin{proof}
Given that $i \in N_{\mathcal{G_T}}(j,2)$ and $i\notin N_{\mathcal{G_T}}(j)$. Using (\ref{PhiInv_expression})
\begin{align*}
&{\boldsymbol{B}_{j}^{\prime}}{{\Phi}_{X}^{-1}(\omega)}{\boldsymbol{B}_{i}}=-{{\Phi}_{E_j}^{-1}(\omega)}{\mathsf{H}_{ji}(\omega)}-({\mathsf{H}_{ij}(\omega)})^*{{\Phi}_{E_i}^{-1}(\omega)}+\\
&\ \ \ \ \ \ \ \ \ \ \ \ \ \ \ {\sum_{k \in \mathcal{C}_G(j)\cap \mathcal{C}_G(i)}}(\mathsf{H}_{kj}(\omega))^*{{\Phi}_{E_k}^{-1}(z)}\mathsf{H}_{ki}(\omega)\\
&= -\boldsymbol{0} -\boldsymbol{0} +{\sum_{k \in \mathcal{C}_G(j)\cap \mathcal{C}_G(i)}}{(\mathsf{H}_{kj}(\omega))^{*}}{{\Phi}_{E_k}^{-1}(\omega)}\mathsf{H}_{ki}(\omega)\\
&={\sum_{k \in \mathcal{C}_G(j)\cap \mathcal{C}_G(i)}}[D(\frac{\omega}{T})\mathsf{h}_{kj}(\frac{\omega}{T})]^*{{\Phi}_{E_k}^{-1}(\omega)}D(\frac{\omega}{T})\mathsf{h}_{ki}(\frac{\omega}{T}). \\
&\text{ Using }(\ref{phiE_phiP})\\
&={\sum_{k \in \mathcal{C}_G(j)\cap \mathcal{C}_G(i)}}{D^*(\frac{\omega}{T})\frac{b_{kj}}{\mathsf{S^{*}_k}(\frac{\omega}{T})}}{{\Phi}_{P_k}^{-1}|\mathsf{S_k}(\frac{\omega}{T})|^{2}}D(\frac{\omega}{T})\frac{b_{ki}}{\mathsf{S_k}(\frac{\omega}{T})}\\
&={\sum_{k \in \mathcal{C}_G(j)\cap \mathcal{C}_G(i)}}{(\frac{b_{kj}b_{ki}}{|\mathsf{S_k}(\frac{\omega}{T})|^2})}{|\mathsf{S_k}(\frac{\omega}{T})|^2}[D^*(\frac{\omega}{T}){\Phi}_{P_k}^{-1}(\omega)D(\frac{\omega}{T})]    \\
&={\sum_{k \in \mathcal{C}_G(j)\cap \mathcal{C}_G(i)}}{(b_{kj}b_{ki})}[D^*(\frac{\omega}{T}){\Phi}_{P_k}^{-1}(\omega)D(\frac{\omega}{T})]\\
\end{align*}

It follows from (\ref{phiE_phiP}) that ${{\Phi}_{P_k}}(\omega)$ is also positive definite, thus $
{{\Phi}_{P_k}^{-1}}(\omega) \succ 0$ and $ [D^*(\frac{\omega}{T}){\Phi}_{P_k}^{-1}(\omega)D(\frac{\omega}{T})]\succeq 0$. 

For a common child $k \in \mathcal{C}_G(j)\cap \mathcal{C}_G(i)$, the values of $b_{kj},b_{ki}$ strictly positive. Thus,
\small
\begin{align*}
{\boldsymbol{B}_{j}^{\prime}}{{\Phi}_{X}^{-1}(\omega)}{\boldsymbol{B}_{i}}&\\
={\sum_{k \in \mathcal{C}_G(j)\cap \mathcal{C}_G(i)}}&{(b_{kj}b_{ki})}[D^*(\frac{\omega}{T}){\Phi}_{P_k}^{-1}(\omega)D(\frac{\omega}{T})] \succeq 0, \\ 
&\ \ \ \ \forall \omega \in [0, 2 \pi).
\end{align*}
\normalsize

\begin{comment}
Using (\ref{Phi_wji}), it follows that,
$\boldsymbol{\mathsf{W}}_{ji} \preceq 0, \forall \omega \in [0, 2 \pi)$.
\end{comment}
\end{proof}

\begin{remark}
The consequence of Theorem \ref{3.1} is that if $i$ is a strict two hop neighbor of $j$, then all the eigenvalues of the  matrix ${\boldsymbol{B}_{j}^{\prime}}{{\Phi}_{X}^{-1}(\omega)}{\boldsymbol{B}_{i}}$ are non-negative. This theorem does not guarantee on its converse, but such cases are pathological as it will be evident in next theorem. 
\end{remark}
\begin{theorem}\label{3.2}
Given a well-posed and topologically detectable LDM $(\mathbb{H}(z),{E})$ with associated graph $\mathcal{G}$ and its topology $\mathcal{G_T}$, the following holds:
\begin{enumerate}
 \item Suppose nodes $i$ and $j$ are such that $i \in N_{\mathcal{G_T}}(j)$, $i\not\in N_{\mathcal{G_T}}(j,2)$, ${\boldsymbol{B}_{j}^{\prime}}{{\Phi}_{X}^{-1}}{\boldsymbol{B}_{i}} \succeq 0$ for all $\omega \in [0, 2\pi)$. Then
 
 $[ -{\frac{b_{ji}}{{\mathsf{S}_j^*}(\frac{\omega}{T})}{\Phi}_{P_j}^{-1}(\omega)}D(\frac{\omega}{T})-\frac{b_{ij}}{\mathsf{S}_i(\frac{\omega}{T})}D^*(\frac{\omega}{T}){{\Phi}_{P_i}^{-1}(\omega)}] \succeq 0$
 \item Suppose $i$ and $j$ are such that $i \in N_{\mathcal{G_T}}(j)$ and $i \in N_{\mathcal{G_T}}(j,2)$ (both one hop and two hop neighbor), with ${\boldsymbol{B}_{j}^{\prime}}{{\Phi}_{X}^{-1}}{\boldsymbol{B}_{i}} \succeq 0$ for all $\omega \in [0,2\pi)$. Then $[ -{\frac{b_{ji}}{{\mathsf{S}_j^*}(\frac{\omega}{T})}{\Phi}_{P_j}^{-1}(\omega)}D(\frac{\omega}{T})-\frac{b_{ij}}{\mathsf{S}_i(\frac{\omega}{T})}D^*(\frac{\omega}{T}){{\Phi}_{P_i}^{-1}(\omega)}+ {\sum_{k \in \mathcal{C}_G(j)\cap \mathcal{C}_G(i)}}{(b_{kj}b_{ki})}[D^*(\frac{\omega}{T}){\Phi}_{P_k}^{-1}(\omega)D(\frac{\omega}{T})]] \succeq 0$
\end{enumerate}
\end{theorem}

\begin{remark}
The consequence of theorem \ref{3.2} is that for nodes $i$ and $j$ that are neighbors but not two hop neighbors, or, nodes $i$ and $j$ that are neighbors and two hop neighbors, the inequalities in theorem \ref{3.2} holds for all $\omega \in [0,2\pi)$ when the system parameters satisfy a restrictive and specific set of conditions. In other words, aside for pathological cases, the converse of Theorem \ref{3.2} holds. So if the ${\boldsymbol{B}_{j}^{\prime}}{{\Phi}_{X}^{-1}}{\boldsymbol{B}_{i}} \succeq 0$ then $i \in N_{\mathcal{G_T}}(j,2)$. We use this condition to prune out spurious two hop neighbor edges using the algorithm discussed next. 

\end{remark}
\subsection{Topology Learning Algorithm}
The steps involved to unravel the structure of a LDG of cyclostationary process are summarized in Algorithm table $1$. 

\begin{comment}
In practice, the (\ref{scalarOpt}) is implemented by allowing lags upto order $F$ and is formulated below
\small
\begin{align}\label{scalarOpt_reg}
 {\inf_{{\{[\boldsymbol{W_{ji}]_{q,:}}\}}_{i=1,..,m,i\neq j}}}\mathbb{E}\big(x_j(k+q-1)-
	& \sum_{i=1,i\neq j}^{m}[\boldsymbol{[W_{ji}]_{q,:}}*X_i(k)]\big)^2\nonumber\\
\end{align}
\normalsize
where $[\boldsymbol{W_{ji}}]_{pq}=[w_{ji,pq}^{-F},...,w_{ji,pq}^{0},...,w_{ji,pq}^{F}]$, 
 $[\boldsymbol{\mathsf{W_{ji}}(z)}]_{pq} = \sum_{L=-F}^{F}w_{ji,pq}^{L}z^{-L}$. 
\end{comment}

\begin{comment}
\begin{remark}
Note that this algorithm can also be applied by lifting the cyclostationary processes to $nT$-variate WSS process, where $n \in \mathbb{N}$. In particular, when the sample size is low, lifting to $nT$-variate WSS process has the potential to improve the algorithm robustness and accuracy of the topology reconstruction. However, the computational complexity increases due to the increase in the size of matrices ($nT\times nT$ instead of $T\times T$) involved. The involved trade-off will be examined in subsequent work. 
\end{remark}
\end{comment}

\begin{algorithm}
\caption{Learning Algorithm for reconstructing the topology of LDG with cyclostationary inputs}
\textbf{Input:} Nodal time series $x_i(k)$ for each node $i \in \{1, 2,... m\}$ which is WSCS. Thresholds $\rho,\tau$. Frequency points $\Omega$. \\
\textbf{Output:} Reconstruct the true topology with an edge set ${\cal E_T}$ \\
\begin{algorithmic}[1]
\State Perform a periodogram analysis of the time series data to determine the period $T$. Arrange each cyclostationary time series $x_i(k)$ as vector time series $X_i(k)$ of size $T$
\State Define $X(k) = [X_1(k),\cdots X_m(k)]^{'}$
\ForAll{$l \in \{1,2,...,m\}$}\label{step1_a}
\State Compute $\boldsymbol{B}_{l}^{\prime}{{\Phi}_{X}^{-1}}{\boldsymbol{B}_{p}}$ using the nodal time series $\forall p \in \{1,2,...,m\} \setminus {l}$ \label{step1_a1}
\EndFor
\State Edge set $\bar{\cal E}_K \gets \{\}$ 
\ForAll{$l,p \in \{1,2,...,m\}, l\neq p$}
\If{$\|\boldsymbol{B}_{p}^{\prime}{{\Phi}_{X}^{-1}}{\boldsymbol{B}_{l}}\|_{\infty} > \rho$}\label{step_nonzero}
\State $\bar{\mathcal{E}}_K \gets \bar{\mathcal{E}}_K \cup \{(l,p)\}$
\EndIf
\EndFor\label{step1_b}
\State Edge set $\bar{\cal E} \gets \bar{\cal E}_K$ \label{step2_a}
\ForAll{$l,p \in \{1,2,...,m\}, l\neq p$}
\State Compute the eigenvalues $\{{\lambda _{pl}(t)\}}^T_{t=1}$ of the matrix $\boldsymbol{B}_{l}^{\prime}{{\Phi}_{X}^{-1}}{\boldsymbol{B}_{p}}$
\If{$\lambda _{pl}(t)$ $\geq -\tau, \forall \omega \in \Omega, \forall t$}
\State $\bar{\mathcal{E}} \gets \bar{\mathcal{E}} - \{(l,p)\}$
\EndIf
\EndFor \label{step2_b}
\end{algorithmic}
\end{algorithm}

\section{Results}
\subsection{Data generation}
In this section, a generative graph is considered as shown in Fig.(\ref{generativegraph}) along with its moral graph and topology. Each node $i$ is driven by exogenous input $e_i(k)$. The processes $\{e_2(k),e_3(k)\}$ are wide sense stationary and $e_1(k)$ is cyclostationary with $T=2$. There are two directed edges in the generative graph where nodes $\{1,3\}$ are the parents of node $2$. The dynamics of the directed edge from node $1$ to node $2$ is given by $\mathsf{h}_{21}(z) = 0.1\{1+0.9z^{-1}+0.5z^{-2}+0.3z^{-3}\}$. Similarly, $\mathsf{h}_{23}(z) = 0.2\{1+0.9z^{-1}+0.5z^{-2}+0.3z^{-3}\}$. 
The data is generated with the following generative model:
\begin{align}\label{eqn:genmodel}
    x_1(k) &= e_1(k)\\
    x_2(k) &= \mathsf{h}_{21}*x_1(k)+\mathsf{h}_{23}*x_3(k)\\
    x_3(k) &= e_3(k)
\end{align}
The sample size of each time series $\{x_1(k),x_2(k),x_3(k)\}$ is $3\times10^5$. 

\begin{figure}[tb]
	\centering
	\begin{tabular}{c}
	\includegraphics[width=0.8\columnwidth, height = 5 cm]{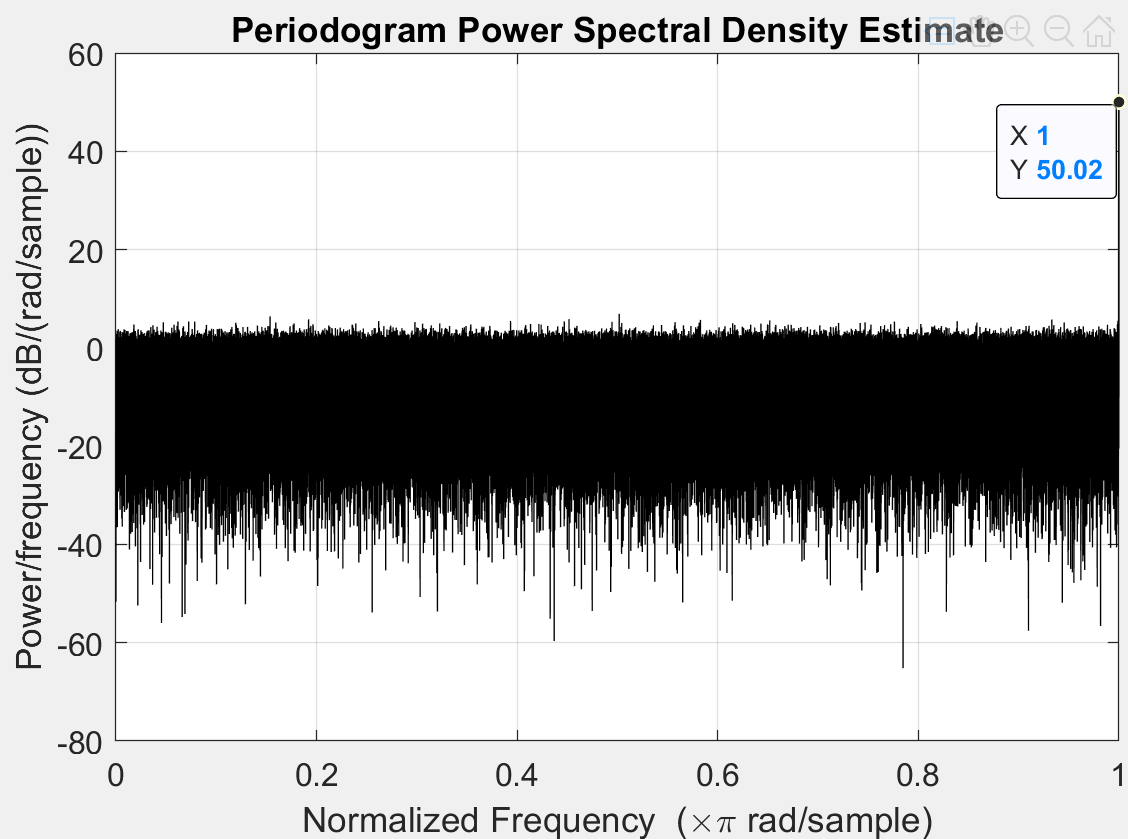}\\
	\end{tabular}
	\caption{Periodogram of time series $x_1(k)$ (peaks at $\pi$, that is $T=2$).
		\label{signal_periodogram}}
\end{figure}
\begin{figure}[tb]
	\centering
	\begin{tabular}{c}
	\includegraphics[width=0.9\columnwidth]{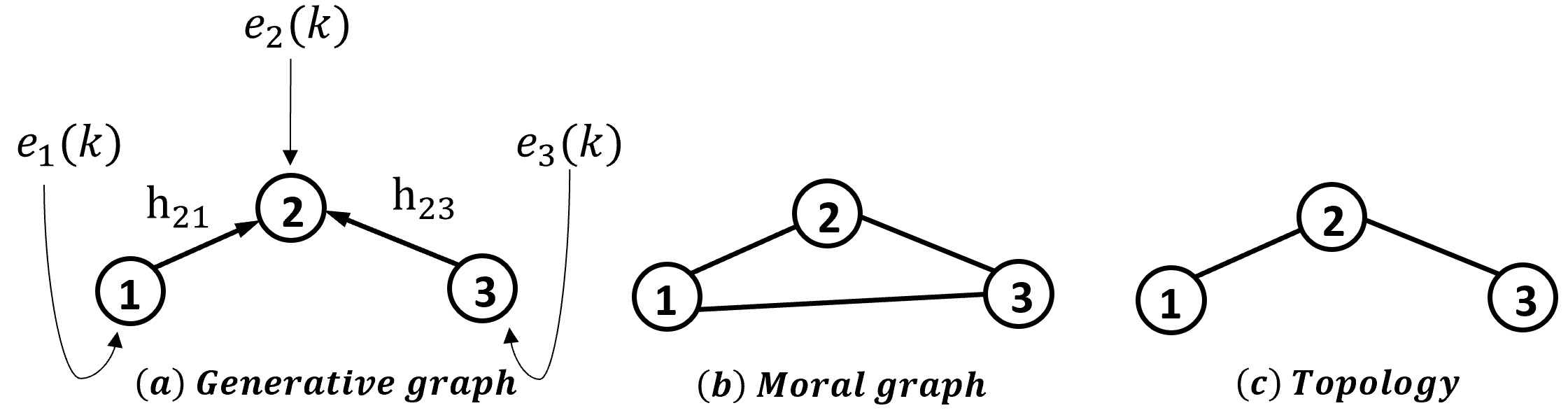}\\
	\end{tabular}
	\caption{(a) Generative graph, (b) Moral graph, (c) Topology.	\label{generativegraph}}
\end{figure}

\begin{figure}[tb]
	\centering
	\begin{tabular}{c}
	\includegraphics[width=1.0\columnwidth]{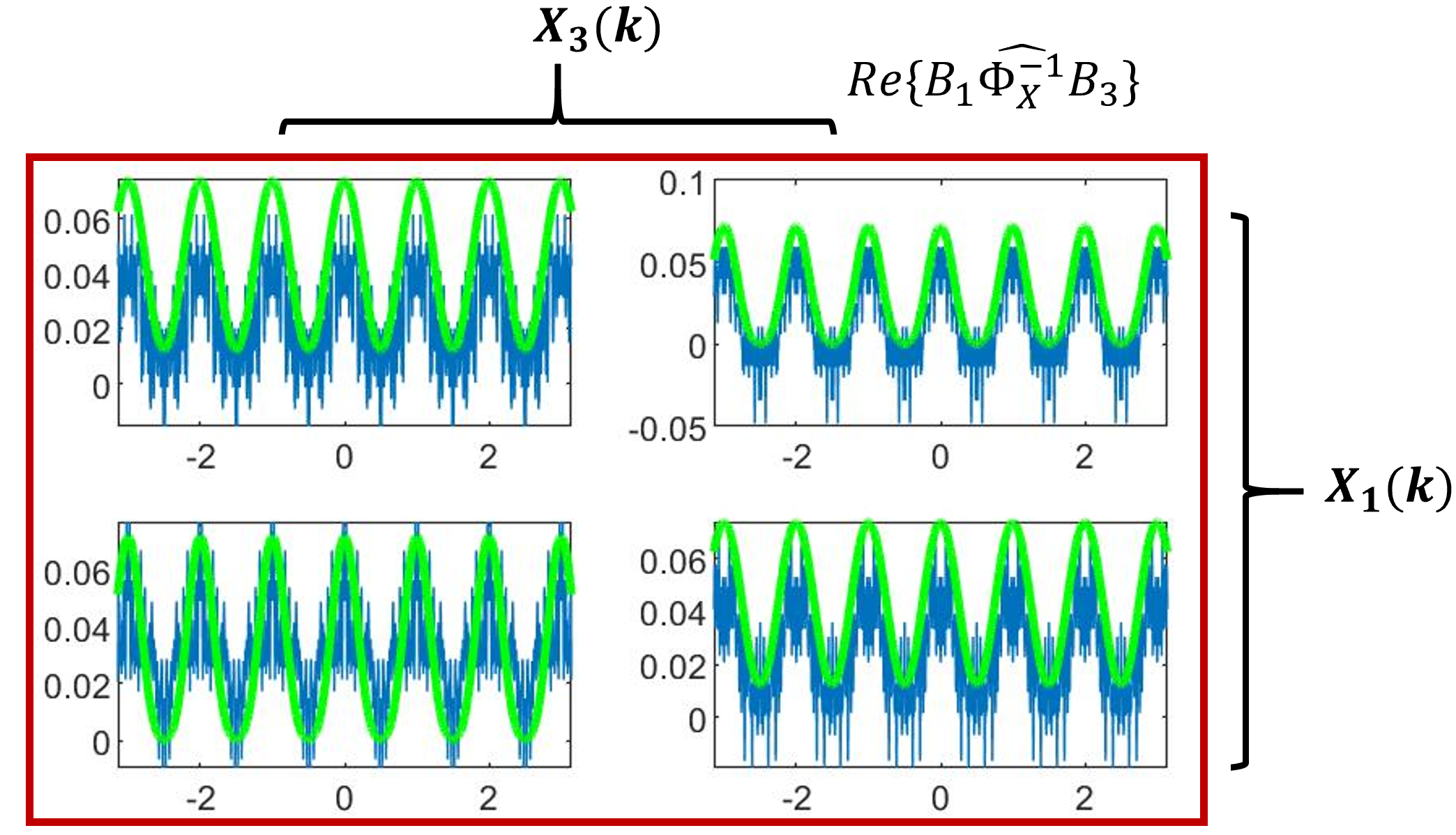}\\
	\end{tabular}
	\caption{Estimate of $\operatorname{Re}({\boldsymbol{B}_{1}^{\prime}}{\hat{\Phi}_{X}^{-1}}{\boldsymbol{B}_{3}})$ (blue) vs frequency $f$. Since we generated the data, the analytical expressions for ${\boldsymbol{B}_{1}^{\prime}}{{\Phi}_{X}^{-1}}{\boldsymbol{B}_{3}}$ is known and plotted in green.	\label{inv_psd}}
\end{figure}

\begin{figure}[tb]
	\centering
	\begin{tabular}{c}
	\includegraphics[width=0.9\columnwidth]{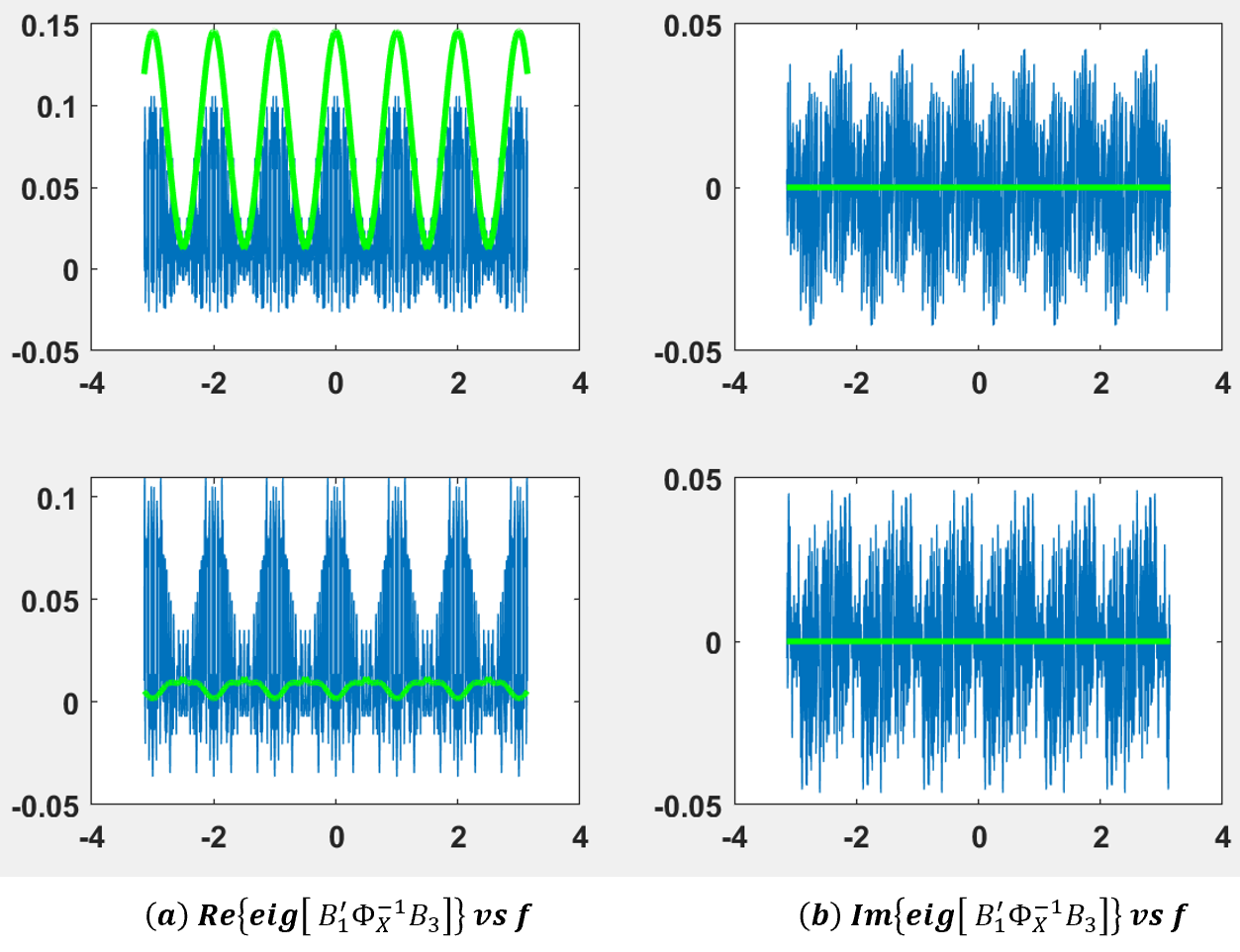}\\
	\end{tabular}
	\caption{Eigenvalues of ${\boldsymbol{B}_{1}^{\prime}}{\hat{\Phi}_{X}^{-1}}{\boldsymbol{B}_{3}}$ vs frequency $f$. Green curve corresponds to \textcolor{green}{true solution} and Blue for \textcolor{blue}{estimated} values using MATLAB.	\label{eigenvalues}}
\end{figure}

\subsection{Topology Reconstruction}
With the data generated in (\ref{eqn:genmodel}), we can apply the presented algorithm to reconstruct the Topology of the Generative graph. The steps involved in reconstructing the Topology are as follows:
\begin{enumerate}
    \item {\it{Identify the time period T}}: Performing the periodogram analysis of the time series $\{x_1(k),x_2(k),x_3(k)\}$, we found the value of $T$ as $2$. For example, the periodogram of $x_1(k)$ is shown in Fig.(\ref{signal_periodogram}).
    \item {\it{Lifting the time series}}: Each time series $x_i(k)$ is lifted to a vector time series $X_i(k)$ of length $T$. That is, $X_i(k)=[x_i(2k), x_i(2k+1)]^T$ for $i=\{1,2,3\}$.
    \item {\it{Compute ${\hat{\Phi}_{X}^{-1}(\omega)}$}}: Using the Welch method \cite{welch1967use}, an estimate of ${{\Phi}_{X}^{-1}(\omega)}$ is computed. This can be done in MATLAB using {\it{cpsd()}}. For example, ${\boldsymbol{B}_{1}^{\prime}}{\hat{\Phi}_{X}^{-1}}{\boldsymbol{B}_{3}}(f)$ is $2\times 2$ matrix and is shown in Fig.\ref{inv_psd}. 
    \item {\it{Moral graph reconstruction}}: We place an undirected edge between two distinct nodes $j$ and $i$ if ${\boldsymbol{B}_{j}^{\prime}}{{\Phi}_{X}^{-1}}{\boldsymbol{B}_{i}}\neq  0$. Repeating this process for all pairs of nodes, the graph that is constructed matches with the moral graph shown in Fig.\ref{generativegraph}(b).
    \item {\it{Topology reconstruction}}: For every edge $(j,i)$ in the reconstructed moral graph, we remove this edge if ${\boldsymbol{B}_{j}^{\prime}}{{\Phi}_{X}^{-1}}{\boldsymbol{B}_{i}}\succeq  0$. The eigenvalues of ${\boldsymbol{B}_{j}^{\prime}}{{\Phi}_{X}^{-1}}{\boldsymbol{B}_{i}}$ are shown in Fig. \ref{eigenvalues}.
    The graph obtained after iterating over all the edges in the moral graph is identical to the topology of the generative graph in Fig.\ref{generativegraph}(c). 
\end{enumerate}
This completes the topology reconstruction of the times series $\{x_1(k),x_2(k),x_3(k)\}$. Since we generated data, we calculated the exact expression of ${{\Phi}_{X}^{-1}(\omega)}$ and compared with the estimates in Fig.\ref{inv_psd} and Fig.\ref{eigenvalues}.

\section{Conclusions}
In summary, the topology of the Linear Dynamic Graph of a Linear Dynamic Model with cyclostationary inputs is constructed from the time series data with provable guarantees. It is a data driven approach useful for exact identification of network structure with applications to power grids, thermal networks, network of rotating mechanical systems amongst many others. Our approach doesn't impose any structural restrictions on the network topology and doesn't use any knowledge of the system parameters. 

In our future work, we will apply this algorithm on experimental data with a wide test cases for the parameters. Regularization methods will be used to improve the performance with less samples and to promote the sparseness. This algorithm can also be used to learn the topology of a network of WSS processes more efficiently than the current state of the art and will be the focus of our future work.
%\bibliography{output.bbl}
%\bibliography{biblio,topident}

\section{Appendix}
Lemma 2.1

\footnotesize
\begin{proof}
The dynamics of a network of cyclostationary processes are given by
\begin{align*}
\mathsf{x}_i(z) &= \sum_{j=1}^{m}\mathsf{h_{ij}}(z)\mathsf{x}_j(z) + \mathsf{e}_i(z)\\
{x}_i(k) &= \sum_{j=1}^{m}\sum_{n= -\infty}^{\infty}h_{ij}(n){x}_j(k-n) + {e}_i(k)\\
{x}_i(kT) &= \sum_{j=1}^{m}\sum_{n= -\infty}^{\infty}h_{ij}(n){x}_j(kT-n) + {e}_i(kT)\\
{x}_i(kT+p) &= \sum_{j=1}^{m}\sum_{n= -\infty}^{\infty}h_{ij}(n){x}_j(kT+p-n) + {e}_i(kT+p)\\
\end{align*}

where $k$ is the time index and $p\in \{0,1,\dots,T-1\}$. Substitute $n=aT+b$, where $b$ takes the values $\{p,p-1,p-2,\dots,p-T+1\}$
\begin{align*}
{x}_i(kT+p) &= \sum_{j=1}^{m}\sum_{n= -\infty}^{\infty}h_{ij}(n){x}_j(kT+p-n) + {e}_i(kT+p)\\
{x}_i(kT+p) &= \sum_{j=1}^{m}\sum_{a= -\infty}^{\infty}\sum_{b=p-T+1}^{p} h_{ij}(aT+b){x}_j(kT+p-[aT+b]) + {e}_i(kT+p)\\
{x}_i(kT+p) &= \sum_{j=1}^{m}\sum_{b=p-T+1}^{p}\sum_{a= -\infty}^{\infty} h_{ij}(aT+b){x}_j(kT+p-[aT+b]) + {e}_i(kT+p)\\
{x}_i(kT+p) &= \sum_{j=1}^{m}\sum_{b=p-T+1}^{p}\sum_{a= -\infty}^{\infty} h_{ij}(aT+b){x}_j(kT+p-b-aT) + {e}_i(kT+p)\\
\text{Put }t = p-b\\
{x}_i(kT+p) &=\sum_{j=1}^{m}\sum_{t=0}^{T-1}\sum_{a= -\infty}^{\infty} h_{ij}(aT+p-t){x}_j([k-a]T+t) + {e}_i(kT+p)\\
\end{align*}
Lets define $H_{ij,pt}[a] =h_{ij}(aT+p-t)$. The value $p$ is the row index starting from $0$ to $T-1$ and $t$ is the column index starting from $0$ to $T-1$ of the filter matrix $H_{ij}$.
Note that for $p=t$ (diagonal entries) $H_{ij,pp}[n] =h_{ij}(nT)$ for $p\in\{0,1,\dots,T-1\}$.
\begin{align*}
{x}_i(kT+p) &=\sum_{j=1}^{m}\sum_{t=0}^{T-1}\sum_{a= -\infty}^{\infty} H_{ij,pt}[a]{x}_j([k-a]T+t) + {e}_i(kT+p)\\
{x}_i(kT+p) &=\sum_{j=1}^{m}[ H_{ij,p0}*{x}_j(kT)+\dots+H_{ij,p\ T-1}*{x}_j(kT+T-1)]+ \\
&\ \ \ {e}_i(kT+p)\nonumber
\end{align*}

Lets lift the scalar process $x_i(k)$ to a vector process $X_i(k)$ by varying $p$ from $0$ to $T-1$. 
\begin{align*}
    X_i(k) = \begin{bmatrix} x_i(kT)\\x_i(kT+1)\\\vdots\\x_i(kT+T-1)  \end{bmatrix}
\end{align*}
Iterate the $p$ in equation (3) from $0$ to $T-1$ to get the following relation

\begin{align*}
    \begin{bmatrix} x_i(kT)\\x_i(kT+1)\\\vdots\\x_i(kT+T-1)  \end{bmatrix}=&\\\nonumber
    \sum_{j=1}^{m} \begin{bmatrix} H_{ij,00} & H_{ij,01}\ \dots \\ \vdots & \ddots & \\ H_{ij,T-1,0}&       & H_{ij,T-1,T-1}
    \end{bmatrix}&* \begin{bmatrix} x_j(kT)\\x_j(kT+1)\\\vdots\\x_j(kT+T-1)  \end{bmatrix}+\\\nonumber
    &\begin{bmatrix} e_i(kT)\\e_i(kT+1)\\\vdots\\e_i(kT+T-1)  \end{bmatrix}\\
\end{align*}
\begin{align}
    X_i(k) =& \sum_{j=1}^{m}H_{ij}*X_j(k)+E_i(k)\\
    \mathsf{X}_{j}(z) =& \sum_{i=1}^{m}\mathsf{H}_{ji}(z)\mathsf{X}_i(z) + \mathsf{E}_j(z)
\end{align}
where $H_{ij,pt}*x_j(kT+t) = \sum_{a=-\infty}^{\infty} H_{ij,pt}[a]x_j(kT+t-aT)$ for $t\in\{0,1,\dots,T-1\}$ with $H_{ij,pt}[n] =h_{ij}(nT+p-t)$.
This implies $\mathsf{H}_{ji}(z) = \mathcal{Z}[H_{ji}(k)] = D(z^{\frac{1}{T}}) \mathsf{h}_{ji}(z^{\frac{1}{T}})$, where\\
\\
$D(z) = \begin{bmatrix}z^0&z^1\dots&z^{T-1}\\\vdots&\ddots&\vdots\\z^{-(T-1)}&\dots&z^0\\\end{bmatrix}$.
\end{proof}

\end{document}